\documentclass[10pt,journal,compsoc]{IEEEtran}

\usepackage{amsmath,amssymb,amsfonts}
\usepackage{algorithmic}
\usepackage{graphicx}
\usepackage{textcomp}
\usepackage{booktabs}
\usepackage{xcolor}
\usepackage{threeparttable}
\usepackage{float}
\usepackage{multirow}
\usepackage{bm}
\usepackage{CJK}
\usepackage{indentfirst}
\usepackage{cases}
\usepackage{ulem}
\newtheorem{definition}{Definition}
\newtheorem{theorem}{Theorem}

\newtheorem{proof}{Proof}
\usepackage[ruled,linesnumbered]{algorithm2e}
\usepackage{hyperref}
\hypersetup{
	colorlinks=true,
	linkcolor=blue,
	citecolor=blue,
	filecolor=magenta,
	urlcolor=cyan,
}
\usepackage{xcolor}

\usepackage{ragged2e}
\begin{document}

\title{Graph Spring Network and Informative Anchor Selection for Session-based Recommendation}

\author{Zizhuo~Zhang,
        and Bang~Wang
\IEEEcompsocitemizethanks{
\IEEEcompsocthanksitem Z. Zhang and B. Wang are with the School of Electronic, Information and Communications, Huazhong University of Science and Technology (HUST), Wuhan, China 430074 \protect\\
E-mail: \{zhangzizhuo, wangbang\}@hust.edu.cn
\IEEEcompsocthanksitem Corresponding Author: Bang Wang.
\IEEEcompsocthanksitem  This work is supported in part by National Natural Science Foundation of China (Grant No: 62172167).}%
\thanks{Manuscript received \today; revised XXX.}}

%

\IEEEtitleabstractindextext{%
\begin{abstract}
\justifying
Session-based recommendation (SBR) aims at predicting the next item for an ongoing anonymous session. The major challenge of SBR is how to capture richer relations in between items and learn ID-based item embeddings to capture such relations. Recent studies propose to first construct an item graph from sessions and employ a Graph Neural Network (GNN) to encode item embedding from the graph. Although such graph-based approaches have achieved performance improvements, their GNNs are not suitable for ID-based embedding learning for the SBR task. In this paper, we argue that the objective of such ID-based embedding learning is to capture a kind of \textit{neighborhood affinity} in that the embedding of a node is similar to that of its neighbors' in the embedding space. We propose a new graph neural network, called Graph Spring Network (GSN), for learning ID-based item embedding on an item graph to optimize neighborhood affinity in the embedding space. Furthermore, we argue that even stacking multiple GNN layers may not be enough to encode potential relations for two item nodes far-apart in a graph. In this paper, we propose a strategy that first selects some informative item anchors and then encode items' potential relations to such anchors. In summary, we propose a GSN-IAS model (\underline{G}raph \underline{S}pring \underline{N}etwork and \underline{I}nformative \underline{A}nchor \underline{S}election) for the SBR task. We first construct an item graph to describe items' co-occurrences in all sessions. We design the GSN for ID-based item embedding learning and propose an \textit{item entropy} measure to select informative anchors. We then design an unsupervised learning mechanism to encode items' relations to anchors. We next employ a shared gated recurrent unit (GRU) network to learn two session representations and make two next item predictions. Finally, we design an adaptive decision fusion strategy to fuse two predictions to make the final recommendation. Extensive experiments on three public datasets demonstrate the superiority of our GSN-IAS model over the state-of-the-art models.
\end{abstract}

}

\maketitle

\IEEEdisplaynontitleabstractindextext

%
\IEEEpeerreviewmaketitle

\IEEEraisesectionheading{\section{Introduction}\label{sec:introduction}}

\IEEEPARstart{R}{ecommendation} systems as an important information filtering device can effectively reduce information overload and assist users easily finding their mostly interested items. A common approach~\cite{su:et.al:2009:AI:survey,shi:et.al:2014:CSUR,zhang:et.al:2019:CSUR} to provide personalized recommendation is through analyzing user profile and his behavior to match his mostly interested items. However, user profile may not be available in many practical situations, for example,  a user browsing items without logging. The task of session-based recommendation (SBR) is hence proposed for such situations, which has attracted widespread attentions recently~\cite{wang:et.al:2019:arXiv:survey}. According to~\cite{li:et.al:2017:CIKM,liu:et.al:2018:KDD,wu:et.al:2019:AAAI}, the SBR task can be defined as follows:

\par
Let $V=\{v_i|i=1,2,...,N\}$ denote the set of all unique items, where $N$ is the total number of items. A session is a chronological sequence of clicked items for an anonymous user, which can be represented as $S=\{v_1, v_2, ..., v_\tau\}$, where $v_i \in V$ in the session is the $i$-th item clicked by the user and $\tau$ is the session length. The SBR task is to predict the next item to be clicked for an ongoing session, i.e. $v_{\tau+1}$ for session $S$. A SBR recommendation model predicts the preference scores for all candidate items in $V$, i.e. $\hat{\mathbf{y}}=\{\hat{y}_1,\hat{y}_2,...,\hat{y}_N\}$, and constructs a top-$K$ recommendation list by selecting $K$ items with highest preference scores.

\par
The challenges of SBR lies in two aspects: (1) Sessions are anonymous and items are recorded with only their identifiers (ID), which makes it impossible to establish users' profiles nor to analyze items' features; (2) Sessions are often short, which makes it difficult to capture transitional logics in sessions, not even mention that items often have no attribution information. Early approaches have mainly relied on some simple popularity rules~\cite{sarwar:et.al:2001:WWW} or matrix factorization~\cite{rendle:et.al:2009:UAI} for predicting the next item, which have not well encoded items' representations, nor utilizing items' sequential relations in sessions~\cite{wu:et.al:2019:AAAI}. Recently, for their powerful capabilities of representation learning, many neural networks have been designed and applied for the SBR task.

\par
Some neural network-based methods~\cite{li:et.al:2017:CIKM,liu:et.al:2018:KDD,wang:et.al:2019:SIGIR,luo:et.al:2020:IJCAI,hidasi:et.al:2016:ICLR,ren:et.al:2019:AAAI} have achieved significant improvements over traditional approaches. The core idea is to design a neural network model for encoding items' representations and learning the ongoing session representation, such that the next item is selected with its representation the most similar to that the session representation, e.g., using a cosine or dot product similarity function. Some of neural models~\cite{wang:et.al:2019:SIGIR,luo:et.al:2020:IJCAI} have further considered the collaborative information between multiple sessions to enhance the session representation learning. However, they lack the exploration of potential relations in between items from multiple sessions.

\par
Some graph-based models propose to first construct an item graph for item embedding learning and session representation learning~\cite{wu:et.al:2019:AAAI,xu:et.al:2019:IJCAI,xia:et.al:2021:AAAI,qiu:et.al:2019:CIKM,chen:et.al:2020:KDD,wang:et.al:2020:SIGIR,zheng:et.al:2020:ICDMW,yu:et.al:2020:SIGIR,pan:et.al:2020:CIKM}. As many items' transitions are extracted from multiple sessions, it is expected that more useful transitional information can be encoded in items' embeddings via, say for example, a \textit{graph neural network} (GNN). We approve the effectiveness of item graph construction, as it has two advantages: One is to alleviate the issue of item sparsity in individual sessions; The other is to transfer the problem of mining item inter-relations as exploring structural characteristics of a graph. Some well-known GNNs, including \textit{graph convolution network} (GCN)~\cite{kipf:et.al:2017:ICLR}, \textit{graph attention network} (GAT)~\cite{velivckovic:et.al:2018:ICLR}, LightGCN~\cite{he:et.al:2020:SIGIR}, have been adopted for learning item embedding mainly from the viewpoint of item co-occurrences in the graph by capturing node-centric local structural characteristics. Although these GNNs can encode graph structural characteristics, they may not be much suitable for item embedding learning in the SBR task.

\par
The items contained in anonymous sessions are only equipped with their identifiers (ID). On the one hand, items' embeddings are often randomly initialized. On the other hand, the objective of  such ID-based embedding learning is to ensure a kind of \textit{neighborhood affinity} by capturing items' co-occurrences as their topological relations on the item graph. We refer the neighborhood affinity to indicate the condition that a node embedding is closer to its neighbors' than its non-neighbors' in the embedding space. Some GNNs, like the GCN and GAT, employ a kind of global transformation kernel to update items' initial embeddings. Although they can capture local structural characteristics, such global transformation may not ensure the neighborhood affinity in the original item embedding space, not to mention introducing additional trainable parameters for embedding space transformation. The LightGCN, though without a global transformation, the aggregation computation simply updates a node embedding via averaging its neighbors', yet not including the node itself embedding in the update process.

\par
In this paper, we propose a new \textit{graph spring network} (GSN) for learning item embedding on an item graph. The basic idea of GSN is to ensure the neighborhood affinity in the embedding space by aggregating the neighbors' information to a node through a series of current optimal aggregation weights in an iterative way. Compared with the GCN and GAT, it does not employ some global transformation kernel, so without involving additional trainable parameters. Compared with the LightGCN, its aggregation weights computation directly depends on item embeddings, rather than simple graph Laplacian matrix.

\par
We can stack multiple GSN layers to include high-order neighbors for learning item embedding, however it still may not be enough to encode possible relations for two item nodes far-apart in the item graph. Indeed, the item graph only establishes edges of co-appeared items in sessions. It could be the case that some items have not co-appeared in sessions, but they may share some similar latent features, say for example categorical features. This resembles the case of two nodes far-apart in the item graph but containing similar latent features, which motivates us to further encode some global topological signals for each item. In this paper, we propose to first select some informative items (called anchors) and then encode potential relations to these anchors.

\par
In this paper, we proposed a GSN-IAS model (\underline{G}raph \underline{S}pring \underline{N}etwork and \underline{I}nformative \underline{A}nchor \underline{S}election) for the SBR task. An item graph is first constructed for describing items' co-occurrences in all sessions. We design the GSN for item embedding learning based on the item graph, and analyze its convergence property. We propose a measure, called \textit{item entropy}, to select informative anchors only based on their statistics in sessions. We next design an unsupervised learning mechanism to output items' encodings, which is to encode items' potential relations to anchors. Based on item embedding and item encoding, we employ a shared \textit{gated recurrent unit} (GRU) network to learn two session representations and output two next item predictions. Finally, we design an adaptive decision fusion strategy to fuse the two predictions to output final preference scores for recommendation list construction. Extensive experiments on three public real-world datasets validate that our GSN-IAS outperforms the state-of-the-art methods .
\par
Our main contributions can be summarized as follows:
\begin{itemize}
	\item Propose a new graph neural network (GSN) for item embedding learning.
	\item Provide a proof on the GSN convergence.
	\item Define item entropy for informative anchor selection.
	\item Propose an unsupervised method for anchor-based item encoding learning.
	\item Propose a parameters shared GRU for two session representations and two next item predictions.
    \item Propose an adaptive decision fusion to fuse two predictions for final recommendation.
\end{itemize}
\par
The rest is organized as follows: Section~\ref{Sec:Retaled Work} reviews the related work. Section~\ref{Sec:Method} presents our GSN-IAS model. Experiment settings are provided in Section~\ref{Sec:Experiment Settings} and results are analyzed in Section~\ref{Sec:Experiment Results}. Section~\ref{Sec:Conlusion} concludes this paper.

\begin{figure*}[t]
	\centering
	\includegraphics[width=\textwidth]{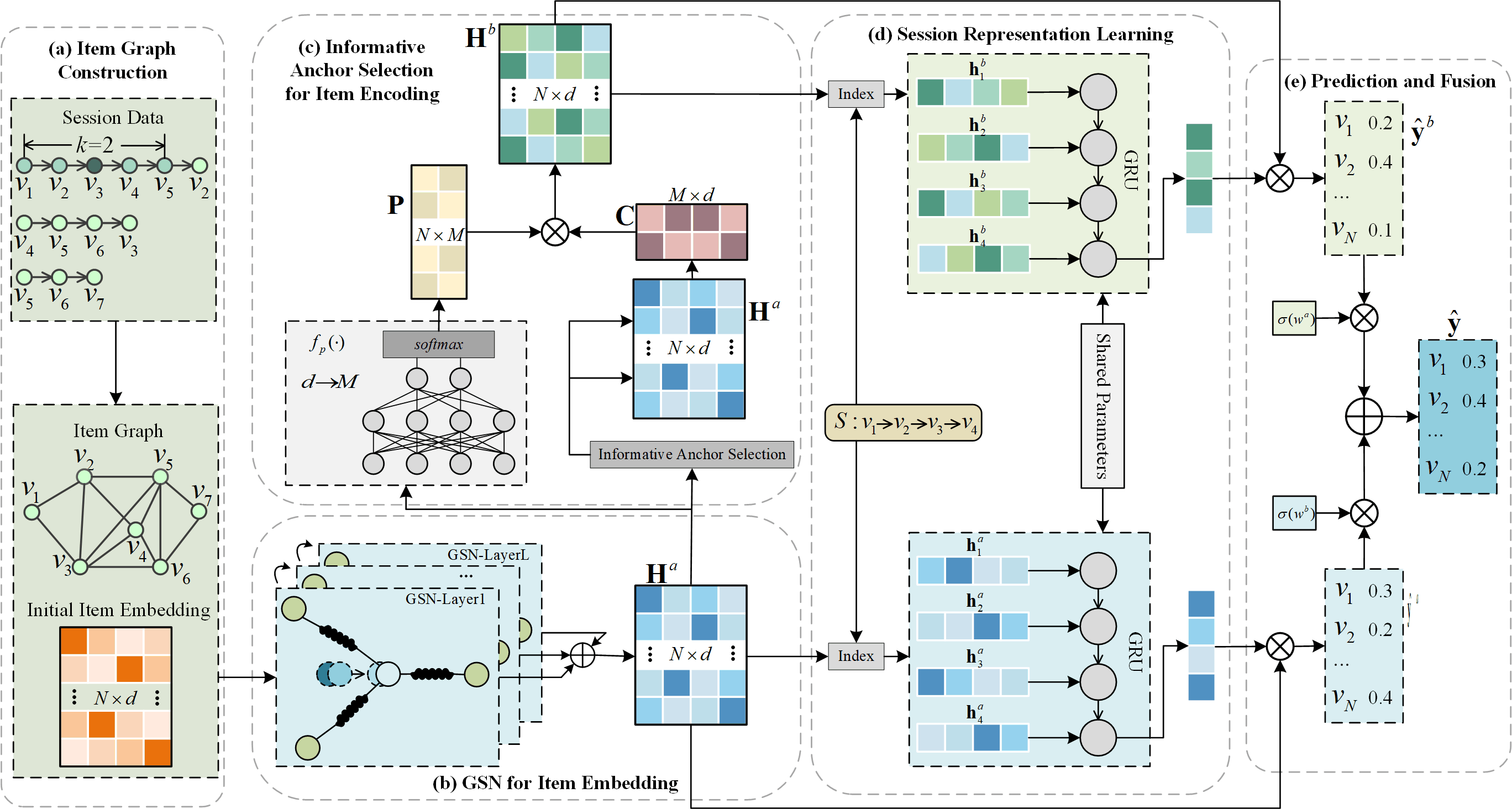}
	\caption{The overall architecture of the proposed GSN-IAS model. (a) Item graph construction; (b) Graph spring network for item embedding; (c) Informative item selection for item encoding; (d) A shared GRU network for session representation learning; (e) Adaptive fusion of two predictions.}
	\label{Fig:Model}
\end{figure*}

\section{Related Work}
\label{Sec:Retaled Work}
We review related work from three aspects: traditional methods, neural network methods and graph-based methods.

\subsection{Traditional Methods}
Some general recommendation approaches~\cite{sarwar:et.al:2001:WWW,  rendle:et.al:2009:UAI, linden:et.al:2003:IC} can be directly applied in the SBR task. For example, the item-based method~\cite{sarwar:et.al:2001:WWW} recommends items by computing items' similarity based on the co-occurrence relations. However, these methods neglect the sequential information in a session, i.e. the order of items. Some other studies~\cite{shani:et.al:2005:JML, rendle:et.al:2010:WWW, zimdars:et.al:2001:UAI} employ Markov chain to model sequential information of a session. For example, Rendle et al.~\cite{rendle:et.al:2010:WWW} combine matrix factorization and first-order Markov chain for next item prediction. However, the independence assumption of Markov chain-based methods is so strict, limiting its prediction accuracy~\cite{wu:et.al:2019:AAAI}. Recently, some studies~\cite{jannach:et.al:2017:Recsys,garg:et.al:2019:SIGIR} realize the importance of items' co-occurrence relations and temporal properties in sessions, and they can achieve competitive performance by using the k-nearest-neighbor approach and temporal decay functions.

\subsection{Neural Network Methods}
Recurrent neural networks (RNNs) are capable of sequential modeling, and many RNN-based models have been designed for the SBR task~\cite{li:et.al:2017:CIKM, wang:et.al:2019:SIGIR, ren:et.al:2019:AAAI,  tan:et.al:2016:DLRS, hidasi:et.al:2016:ICLR, quadrana:et.al:2017:Recsys, pan:et.al:2020:SIGIR:An, zhang:et.al:2020:Neurocomputing}. For example, Hidasi et al.~\cite{hidasi:et.al:2016:ICLR} propose a GRU4Rec model by applying a GRU layer to encoder sequential items. Li et al.~\cite{li:et.al:2017:CIKM} propose a NRAM model to boost a GRU layer with attention to learn session representation by discriminating GRU hidden states. Wang et al.~\cite{wang:et.al:2019:SIGIR} propose a CSRM model to exploit collaborative information from multiple sessions. Ren et al.~\cite{ren:et.al:2019:AAAI} propose a RepeatNet model that incorporates a repeat-explore mechanism into a GRU network to automatically learn the switch probabilities between repeat and explore modes.

\par
Besides using RNNs, some other neural networks have also be explored for the SBR task~\cite{liu:et.al:2018:KDD,luo:et.al:2020:IJCAI, kang:et.al:2018:ICDM, yuan:et.al:2019:WSDM, song:et.al:2019:IJCAI, zhou:et.al:2019:WWW, pan:et.al:2020:SIGIR:Rethinking,yuan:et.al:2021:AAAI,cho:et.al:2021:SIGIR}. For example, Liu et al.~\cite{liu:et.al:2018:KDD} propose the STAMP model that employs an attention network to emphasize the importance of the last click item when learning session representation. Kang et al.~\cite{kang:et.al:2018:ICDM} propose the SASRec algorithm by stacking self-attention layers to capture latent relations in between consecutive items in a session. Yuan et al.~\cite{yuan:et.al:2019:WSDM} propose a NextItNet model by using a dilated convolutional network to capture dependencies in between items. Song et al.~\cite{song:et.al:2019:IJCAI} propose the ISLF model that employs a recurrent variational auto-encoder (VAE) to take interest shift into account. Pan et al.~\cite{pan:et.al:2020:SIGIR:Rethinking} apply a modified self-attention network to estimate items' importance in a session. Luo et al.~\cite{luo:et.al:2020:IJCAI} integrate the collaborative self-attention network to learn the session representation and predict the intent of the current session by investigating neighborhood sessions. Yuan et al.~\cite{yuan:et.al:2021:AAAI} propose a dual sparse attention network to find the possible unrelated item in sessions.

\par
The aforementioned neural models mainly consider item transition or co-occurrence relations within a single session or several similar sessions, which limits their performance because of their lack of item embedding learning against the whole set of all available items and sessions.

\subsection{Graph-based Methods}
Recently, graph neural models have aslo been adopted in the SBR task~\cite{wu:et.al:2019:AAAI, xu:et.al:2019:IJCAI, xia:et.al:2021:AAAI, qiu:et.al:2019:CIKM, chen:et.al:2020:KDD, wang:et.al:2020:SIGIR, zheng:et.al:2020:ICDMW, yu:et.al:2020:SIGIR, pan:et.al:2020:CIKM,xia:et.al:2021:CIKM,zhang:et.al:2021:InfoSci,li:et.al:2022:TKDE}. Wu et al.~\cite{wu:et.al:2019:AAAI} propose a SR-GNN model to construct a session graph and employ the GGNN model~\cite{li:et.al:2015:arxiv} for item representation learning. Xu et al.~\cite{xu:et.al:2019:IJCAI} make a combination of GNN and multi-layer self-attention network. Qiu et al. propose a FGNN model to convert a session into a directed weighted graph and use a weighted attention graph neural layer for item embedding learning. Zheng et al.~\cite{zheng:et.al:2020:ICDMW} propose a DGTN model in which a session together with its similar sessions are used to first construct an item graph. Yu et al.~\cite{yu:et.al:2020:SIGIR} propose a TAGNN model with a target attention mechanism to learn a target-aware session representation. Wang et al.~\cite{wang:et.al:2020:SIGIR} propose a GCE-GNN model to learn two types of item embeddings from a session graph and a global item graph. Zhang et al.~\cite{zhang:et.al:2021:InfoSci} propose a random walk approach to mine a kind of latent categorical information of items from an item graph. Xia et al.~\cite{xia:et.al:2021:AAAI} learn the inter- and intra-session information from two types of hypergraphs. Li et al.~\cite{li:et.al:2022:TKDE} propose a disentangled GNN method to cast item and session embeddings with disentangled representations of multiple factors.

\section{The Proposed GSN-IAS model}\label{Sec:Method}
Fig.~\ref{Fig:Model} presents the architecture of our GSN-IAS model, which consists of the following parts: (1) item graph construction; (2) graph spring network for item embedding; (3) informative anchor selection for item encoding; (4) session representation learning; and (5) prediction and fusion.

\subsection{Item Graph Construction} \label{Sec:ItemGraph}
To deal with item sparsity in one session, we first construct an \textit{item graph} to capture potential items' inter-relations from all sessions. In particular, we construct an undirected weighted item graph $\mathcal{G}=\{\mathcal{V}, \mathcal{E}\}$ with $\mathcal{V}$ the set of nodes (viz. items) and $\mathcal{E}$ the set of edges as follows: An edge $e_{ij}=(v_i, v_j, w_{i,j})$ in $\mathcal{E}$ is established between an item $v_i$ and another item $v_j$, if the $v_j$ is within a window of size $k$ centered at $v_i$ in any session, and $w_{i,j}$ is its weight corresponding the number of appearances of the relations $(v_i, v_j)$ in all sessions. Fig.~\ref{Fig:Model} illustrates an item graph constructed from three sessions. Take item $v_3$ for example. For a window
size $k=2$, the neighbors of item $v_3$ include $\{v_1, v_2, v_4, v_5, v_6\}$.

\begin{figure}[t]
	\centering
	\includegraphics[width=0.5\textwidth]{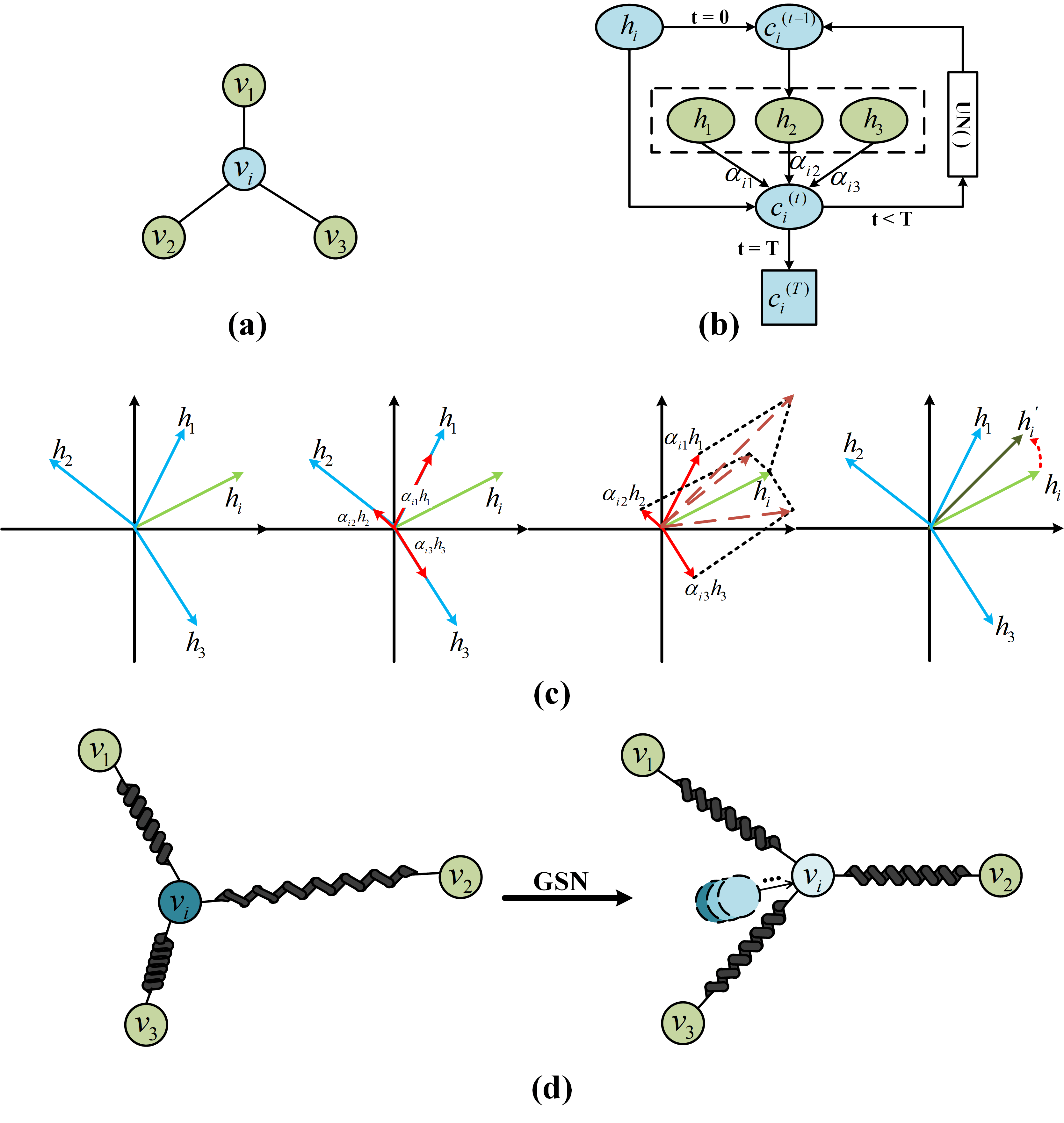}
	\caption{An example to illustrate the GSN core operations. (a) A node $v_i$ with three neighbors $v_1, v_2, v_3$. (b) The iteration process of GSN. (c) The lower part illustrates the input embeddings of the four nodes. After aggregation weights computation, it next illustrates the weighted neighbors' embeddings. The adjustments according to neighbors' embeddings are illustrated as dashed arrows; While the last one shows the updated embedding for node $v_i$ according to its own embedding and the weighted sum of its neighbors' embeddings. (d) Likens the GSN operation to a process in which multiple springs find a balance point via springs' push and pull.}
	\label{Fig:GSN}
\end{figure}

\begin{algorithm}[t]
	\caption{The Graph Spring Network (core operation for one node $v_i$)}
	\label{Alg:GSN}
	\KwIn{$\mathbf{h}_i \in \mathbb{R}^d$ (the embedding of node $v_i$), and $\{\mathbf{h}_j\in \mathbb{R}^d|v_j \in \mathcal{N}_i\}$ (the embeddings of $v_i$'s neighbors)}
	\KwOut{$\mathbf{h}_i^{\prime}\in \mathbb{R}^d$ (updated embedding of $v_i$)}
	$\mathbf{h} \leftarrow \mathtt{UN}(\mathbf{h}) = \frac{\mathbf{h}}{\|\mathbf{h}\|}$, for all $v_i, v_j$ \\
	$\mathbf{c}_i^{(0)}$ $\leftarrow$ $\mathbf{h}_i$ \\
	\For{$t=1, 2, ..., T$}{
		$w = 0$\\
		\For{$v_j \in \mathcal{N}_i$}{
			$\alpha_{ij} \leftarrow \mathbf{h}_j^{\mathsf{T}} \cdot \mathbf{c}_i^{(t-1)}$  // Compute similarity \\
			$w \leftarrow w +\exp(\alpha_{ij})$}
		\For{$v_j \in \mathcal{N}_i$}{
			$\alpha_{ij} \leftarrow \frac{\exp(\alpha_{ij})}{w}$ // Normalize to probability
		}
		$\mathbf{c}_i \leftarrow \mathtt{UN}(\mathbf{h}_i + \sum_{v_j \in \mathcal{N}_i} \alpha_{ij} \mathbf{h}_j)$  // Update
	}
	\KwResult{$\mathbf{h}_i^{\prime} \leftarrow \mathbf{c}_i$ }
\end{algorithm}

\subsection{Graph Spring Network for Item Embedding}

\subsubsection{Graph Spring Network}
We propose a \textit{graph spring network} (GSN) to learn nodes' embeddings from a graph, which aims to find a \textit{balance point} for each node in the embedding space given its neighbors' embeddings in an iterative way. We first introduce the operation of \textit{unit normalization} for regulating a node's embedding as a \textit{unit vector} as follows:
\begin{equation}
	\mathtt{UN}(\mathbf{h}) = \frac{\mathbf{h}}{\|\mathbf{h}\|},
\end{equation}
where $\mathbf{h}$ denotes an item embedding and $\|\mathbf{h}\|$ its vector norm, 2-norm here.

\par
For an item node $v_i$ in the item graph $\mathcal{G}$, let $\mathcal{N}_i$ denote the set of its one-hop neighbors. Given its neighbors' embeddings $\{\mathbf{h}_j|v_j \in \mathcal{N}_i\}$, the core operation of GSN in one iteration contains two main steps:
\begin{align}
	\alpha_{ij}^{(t+1)} & = \frac{\exp(\mathbf{h}_j^\mathsf{T}\cdot \mathbf{c}_i^{(t)})}{\sum_{v_j \in \mathcal{N}_i} \exp(\mathbf{h}_j^\mathsf{T}\cdot \mathbf{c}_i)}, \label{Eq:GSNWeight} \\
	\mathbf{c}_{i}^{(t+1)} & = \mathtt{UN}( \mathbf{h}_i + \sum_{v_j \in \mathcal{N}_i} \alpha_{ij}^{(t+1)} \mathbf{h}_j ) \label{Eq:GSNUpdate},
\end{align}
where $\mathbf{c}_i$ is an intermediate variable used for iteration and $\mathbf{c}_i^{(0)} = \mathbf{h}_i$. Note that the unit normalization obtains a unit vector without changing its direction. Algorithm~\ref{Alg:GSN} presents the pseudo-codes for the core operations of one-layer GSN. Notice that for the input embeddings of a node $v_i$ and its neighbors, one-layer GSN iterates $T$ times to update the node $v_i$'s embedding.

\par
Fig.~\ref{Fig:GSN} uses a toy example to illustrate the embedding update process for one node. The weight $\alpha_{ij}$ measures the normalized similarity between node $v_i$ and its neighbor $v_j$ in the embedding space. The weighted vector $\alpha_{ij}\mathbf{h}_j$ can be regarded as some \textit{virtual force} exerted by $v_j$ to $v_i$, liken to a spring between two nodes for pushing close or pull away two nodes; While all $\alpha_{ij}\mathbf{h}_j$s together with $\mathbf{h}_i$ shall decide the new direction of the updated embedding for $v_i$. An iterated process of such core operations would lead to that the embedding $\mathbf{h}_i$ of node $v_i$ gradually converges to some balance point in the embedding space, liken that the virtual forces by its neighbors reach a balance and no longer change its position.

\subsubsection{Convergence analysis}
The design philosophy of our GSN is to enable the embeddings of neighboring nodes can also reflect the local structural relations between a node and its neighbors in the item graph. We investigate whether these local structural relations have a balanced point, i.e. whether GSN converges after a sufficient number of iterations. We next prove the convergence of GSN by combining GSN and a von Mises-Fisher distribution~\cite{ma:et.al:2019:ICML}. That is, for a node $v_i$, the iterative updating leads to the convergence of its normalized similarities $\alpha_{ij}$ to its neighbors as well as its embedding $\mathbf{h}_i$.


\par
The probability density function of the von Mises-Fisher distribution (vMF)\footnote{https://en.wikipedia.org/wiki/Von\_Mises-Fisher\_distribution} is given by:
\begin{equation}
	f(\mathbf{x};\bm{\mu},\kappa) = C(\kappa)\exp(\kappa\bm{\mu}^\mathsf{T}\mathbf{x}) \propto \exp(\kappa\bm{\mu}^\mathsf{T}\mathbf{x}),
\end{equation}
where $\bm{\mu}$ and $\kappa \geq 0$ are called the mean direction and the concentration parameter, respectively, and $C(\kappa)$ is a constant. The larger value of $\kappa$, the higher concentration of the distribution around the mean direction $\bm{\mu}$.

\par
Let $\mathbf{c}_i$ denote the parameter of vMF distribution. We initialize nodes' embeddings $\mathbf{h}_i$ and $\{\mathbf{h}_j | v_j\in \mathcal{N}_i\}$ by:
\begin{align}
	&\mathbf{h}_i \sim {\rm vMF}(\mathbf{c}_i, 1), \\
	&\mathbf{h}_j \sim {\rm vMF}(\mathbf{c}_i, \alpha_{ij}),
\end{align}
where $\alpha_{ij}$ is the concentration of $\mathbf{h}_j$ around the $\mathbf{c}_i$. With such initializations, the GSN operation for computing $\mathbf{c}_i$ can be converted to estimate the parameter $\mathbf{c}_i$ based on the noisy observation $\mathbf{h}=\{\mathbf{h}_i, \mathbf{h}_j | v_i \vee v_j\in\mathcal{N}_i\}$. The following theorem summarizes the property and convergence of the GSN core operation:

\par
\begin{theorem}
	\label{Theo:ConGSN}
	The computation process of GSN core operation is equivalent to an expectation-maximization (EM) algorithm estimating parameter of the vMF distribution by maximizing the likelihood probability $P(\mathbf{h}_i|\mathbf{c}_i)$. The convergence of GSN core operation is equivalent to the convergence of the EM algorithm for estimating the parameter of vMF distribution.
\end{theorem}
\begin{proof}
	\label{Proof:ConGSN}
	See Appendix~\ref{Appendix:GSNProof}.
\end{proof}

\par
Theorem~\ref{Theo:ConGSN} indicates that the GSN iteration process can be transformed as using the EM algorithm for parameter estimation. In Proof~\ref{Proof:ConGSN}, we prove that Eq.~\eqref{Eq:GSNWeight} is performing the E-step, and Eq.~\eqref{Eq:GSNUpdate} is performing the M-step. Our GSN in fact is iteratively performing the E-step and M-step to estimate the parameter $\mathbf{c}_i$ based on the given observations $\mathbf{h}=\{\mathbf{h}_i, \mathbf{h}_j | v_i \vee v_j\in\mathcal{N}_i\}$. According to the convergence of EM algorithm, the GSN therefore converges. The convergence of GSN core operation also reflects the difference between our GSN and some widely used GNNs, like LightGCN~\cite{he:et.al:2020:SIGIR} and GAT~\cite{velivckovic:et.al:2018:ICLR}. Given the neighbors' embeddings of a node $v_i$, our GSN is to iteratively find a locally optimal aggregation weights between $v_i$ and its neighbors; While most GNNs compute the aggregation weights only once by an attention function or by direct graph Laplacian matrix. More comparisons with other GNNs are provided in Appendix~\ref{Appendix:GNNComparison}.

\subsubsection{Item Embedding}
We can also stack multiple layers of GSN to include high-order neighbors for learning a node embedding. Let $\mathbf{h}_i^{(l)} (l=1,...,L)$ denote the $l$-th layer embedding of $v_i$ learned by GSN, and $\mathbf{H}^{(l)}$ denote the $l$-th layer embedding matrix of all nodes. For an item node $v_i$, to enjoy some potential relations with its higher orders' neighbors, we compute its final item embedding $\mathbf{h}_i^a$ by
\begin{equation}
	\mathbf{h}_i^a = \mathbf{h}_i^{(0)} + \mathbf{h}_i^{(1)} + ... + \mathbf{h}_i^{(L)}.
\end{equation}
We denote $\mathbf{H}^a \in \mathbb{R}^{N\times d}$ as the item embeddings matrix, where the $i$-th row is the embedding of item $v_i$.

\subsection{Informative Anchor Selection for Item Encoding}

\subsubsection{Informative Anchor Selection}
We would like to select a few of representative items to describe potential categorical features from thousands of items. To this end, we propose a measure, called \textbf{item entropy}, to evaluate the informativeness of an item. The item entropy $H(v_i)$ of $v_i$ is defined by
\begin{equation}
	H(v_i) = -\sum_{j=1}^{\lvert \mathcal{S} \rvert}P_{v_i, s_j} \cdot \log(P_{v_i, s_j})
\end{equation}
where $\lvert \mathcal{S} \rvert$ is the number of all available sessions. If an item $v_i$ is clicked in a session $s_j$, then $P_{v_i, s_j}$ is computed by
\begin{equation}
\begin{split}
	P_{v_i, s_j} = \frac{\text{number of clicks in session } s_j}{\text{total number of clicks}} \times \\
	\frac{\text{number of sessions including item } v_i}{\text{total number of sessions}};
\end{split}
\end{equation}
Otherwise, $P_{v_i, s_j} = 0$. According to the definition, an item appearing in more sessions has higher item entropy, and a session including more clicks has more influence on item entropy. A large item entropy indicates that an item appears in a lot of sessions and these sessions include more clicks. In this regard, we select $M$ items with the top-$M$ highest item entropy as informative items, called \textbf{anchors}. Let $\mathcal{A}$ denote the set of all selected anchors.

\subsubsection{Item Encoding}
We extract the corresponding rows from item embedding matrix $\mathbf{H}^a$ to construct an anchor embedding matrix $\mathbf{A} \in \mathbb{R}^{M \times d}$. Then we leverage a linear transformation to transfer the anchor embeddings from the item embedding space to another so-called \textit{anchor encoding space}:
\begin{equation}
	\mathbf{C} = \mathbf{W}_c \mathbf{A} + \mathbf{b}_c,
\end{equation}
where $\mathbf{C} \in \mathbb{R}^{M \times d}$ is the embedding matrix of anchors, $\mathbf{W}_c$ and $\mathbf{b}_c$ are trainable linear transformation parameters.

\par
We propose an unsupervised strategy to learn a soft assignment as the item-anchor distribution to describe their relations. Given an embedding $\mathbf{h}_i^a$ of item $v_i$, we use an encoder network to produce the logits of the probabilities, and then convert them to a distribution $\mathbf{p}_i \in \mathbb{R}^M$ by a softmax layer as follows:
\begin{align}
	\boldsymbol{\beta}_i & = f_p(\mathbf{h}_i^a), \\
	p_{i,j} & = \frac{\exp(\beta_{i,j})}{\sum_{j^\prime=1}^{M}\exp(\beta_{i,j^\prime})}, \\
	\mathbf{p}_i & = [p_{i,1}, p_{i,2}, ..., p_{i,M}],
\end{align}
where $f_p(\cdot)$ is the encoder network to map an item embedding $\mathbf{h}_i^a \in \mathbf{R}^d$ to a logits $\boldsymbol{\beta}_i \in \mathbb{R}^M$. Any $d\rightarrow M$ encoder network can be used as $f_p$. In this paper, we employ a two-layer feed forward neural network to complete this task:
\begin{equation}
	f_p(\mathbf{h}) = \mathbf{W}_p^{(2)^\top} {\rm LeakyReLU}(\mathbf{W}_p^{(1)^\top} \mathbf{h} + \mathbf{b}_p^{(1)}) + \mathbf{b}_p^{(2)},
\end{equation}
where $\mathbf{W}_p^{(1)} \in \mathbb{R}^{d\times d}, \mathbf{b}_p^{(1)} \in \mathbb{R}^d$ and $\mathbf{W}_p^{(2)} \in \mathbb{R}^{d\times M}, \mathbf{b}_p^{(2)} \in \mathbb{R}^M$ are trainable parameters. We use $\mathbf{P} \in \mathbb{R}^{N \times M}$ to denote the item-anchor distribution matrix.

\par
We assign the representative information of each anchor to items based on the item-anchor distribution to obtain item encoding as follows:
\begin{equation}
	\mathbf{H}^b = \mathbf{P} \mathbf{C},
\end{equation}
where $\mathbf{H}^b \in \mathbb{R}^{N\times d}$ is the item encoding matrix, the $i$-th row $\mathbf{h}_i^b = \mathbf{p}_i\mathbf{C}$ in $\mathbf{H}^b$ is the item encoding for $v_i$. Notice that this process is an end-to-end learning schema, where the model learns how to assign the distribution for anchors of an input item.

\subsection{Session Representation Learning}
For a session $S=\{v_1, v_2, ..., v_{\tau}\}$, in order to capture its sequential information, we employ a \textit{gated recurrent unit} (GRU) network to learn the session representation from two perspectives of the item embeddings $S^a=\{\mathbf{h}_1^a, \mathbf{h}_2^a, ..., \mathbf{h}_{\tau}^a\}$ and item encodings $S^b=\{\mathbf{h}_1^b, \mathbf{h}_2^b, ..., \mathbf{h}_{\tau}^b\}$:
\begin{align}
	(\mathbf{s}_1^a, \mathbf{s}_2^a, ..., \mathbf{s}_{\tau}^a) & = {\rm GRU}(\mathbf{h}_1^a, \mathbf{h}_2^a, ..., \mathbf{h}_{\tau}^a), \\
	(\mathbf{s}_1^b, \mathbf{s}_2^b, ..., \mathbf{s}_{\tau}^b) & = {\rm GRU}(\mathbf{h}_1^b, \mathbf{h}_2^b, ..., \mathbf{h}_{\tau}^b),
\end{align}
where $\mathbf{s}_{\tau}$ is the hidden state of GRU at the timestamp $\tau$. We select the last hidden state $\mathbf{s}_{\tau}^a$ and $\mathbf{s}_{\tau}^b$ as two session representations. We note that it is one GRU network shared for learning both session representations, which can help the GRU to be jointly optimized for both item embeddings and item encodings.

\subsection{Prediction and Fusion}
We propose to use \textit{decision fusion} to obtain the final score prediction. Note that $\mathbf{s}_{\tau}^a$ and $\mathbf{s}_{\tau}^b$ encode a session $S$ from two different feature spaces: The objective of $\mathbf{H}^a$ is to encode local structural characteristics; While that of $\mathbf{H}^b$ is to encode global topological information. A feature fusion might confuse the two design objectives; yet a decision fusion may be able to reconcile them.

\par
We first independently compute two intermediate predictions $\mathbf{\hat{y}}^a \in \mathbb{R}^d$ and $\mathbf{\hat{y}}^b \in \mathbb{R}^d$  based on $\mathbf{s}_{\tau}^a$ and $\mathbf{s}_{\tau}^b$ by
\begin{align}
	\mathbf{\hat{y}}^a &= {\rm softmax}(\mathbf{H}^a \mathbf{s}_\tau^a), \\
	\mathbf{\hat{y}}^b &= {\rm softmax}(\mathbf{H}^b \mathbf{s}_\tau^b).
\end{align}
We further employ two trainable weights to implement an \textit{adaptive decision fusion} as follows:
\begin{equation}
	\mathbf{\hat{y}} = \sigma(\omega^a) \cdot \mathbf{\hat{y}}^a + \sigma(\omega^b) \cdot \mathbf{\hat{y}}^b,
\end{equation}
where $\mathbf{\hat{y}} \in \mathbb{R}^N$ is the final score prediction for recommendation, and $\sigma(\omega^a), \sigma(\omega^b)$ are the adaptive fusion weights, $\sigma(\cdot)$ is the $\mathtt{sigmoid}$ function. The top-$K$ highest-scored items are selected to construct a recommendation list.

\par
We leverage the cross entropy loss function to supervise the prediction $\mathbf{\hat{y}}^a, \mathbf{\hat{y}}^b$ and $\mathbf{\hat{y}}$ as follows:
\begin{align}
	\mathcal{L}_a(\mathbf{\hat{y}}^a, \mathbf{y}) &= -\sum_{i=1}^{N} y_i \log(\hat{y}_i^a), \\
	\mathcal{L}_b(\mathbf{\hat{y}}^b, \mathbf{y}) &= -\sum_{i=1}^{N} y_i \log(\hat{y}_i^b), \\
	\mathcal{L}_c(\mathbf{\hat{y}}, \mathbf{y}) &= -\sum_{i=1}^{N} y_i \log(\hat{y}_i), \\
	\mathcal{L} &= \mathcal{L}_a + \mathcal{L}_b + \mathcal{L}_c,
\end{align}
where $\mathbf{y}$ is the ground truth, a one-hot indicative vector: $y_i=1$ if $v_i = v_{\tau+1}$; Otherwise $y_i=0$. Similarly, $\hat{y}_i^a, \hat{y}_i^b, \hat{y}_i$ are the $i$-th scores in $\mathbf{\hat{y}}^a, \mathbf{\hat{y}}^b, \mathbf{\hat{y}}$, respectively. At last, we optimize the loss function $\mathcal{L}$ to train our model.

\begin{table}[t]
	\centering
	\caption{Statistics of the three datasets.}
	\begin{tabular}{lccc}
		\toprule
		Datasets & Yoochoose 1/64 & RetailRocket & Diginetica \\
		\midrule
		number of nodes $N$ & 17,376 & 36,968 & 43,097 \\
		number of edges $|E|$ & 227,205 & 542,655 & 782,655 \\
		\bottomrule
	\end{tabular}
	\label{Tble:Graph}
\end{table}

\section{Experiment Settings}
\label{Sec:Experiment Settings}
\subsection{Datasets}
We conduct experiments on three real-world datasets: Yoochoose\footnote{http://2015.recsyschallenge.com/challenge.html}, RetailRocket\footnote{https://www.kaggle.com/retailrocket/ecommerce-dataset} and Diginetica\footnote{http://cikm2016.cs.iupui.edu/cikm-cup}, which are commonly used in the SBR task~\cite{li:et.al:2017:CIKM, liu:et.al:2018:KDD, wu:et.al:2019:AAAI, wang:et.al:2019:SIGIR, qiu:et.al:2019:CIKM, zhang:et.al:2020:Neurocomputing, yuan:et.al:2021:AAAI}. The Yoochoose is from the Recsys Challenge 2015,  containing six months of clicks from an European e-commerce website. The RetailRocket is a Kaggle contest dataset published by an e-commerce company, containing browsing activities in six months. The Diginetica comes from CIKM Cup 2016, and only its transactional data are used. 

\par
To make a fair comparison, our data preprocessing is the same as that of~\cite{li:et.al:2017:CIKM, liu:et.al:2018:KDD,wu:et.al:2019:AAAI,wang:et.al:2019:SIGIR,yuan:et.al:2021:AAAI}. In particular, we filter out sessions of length one and items appearing less than five times. Furthermore, after constructing the item graph, we adopt the data augmentation approach in~\cite{wu:et.al:2019:AAAI, qiu:et.al:2019:CIKM}. Table~\ref{Tble:Dataset} summarizes the statistics of the three datasets.
\begin{table}[t]
	\centering
	\caption{Statistics of the three datasets.}
	\begin{tabular}{lccc}
		\toprule
		Datasets & Yoochoose 1/64 & RetailRocket & Diginetica \\
		\midrule
		\# train sessions & 369,859 & 433,648 & 719,470 \\
		\# test sessions & 55,696 & 15,132 & 60,858 \\
		\# items & 17,376 & 36,968 & 43,097 \\
		\# average lengths & 6.16 & 5.43 & 5.12 \\
		\bottomrule
	\end{tabular}
	\label{Tble:Dataset}
\end{table}

\subsection{Parameter Setting}
Following previous work~\cite{li:et.al:2017:CIKM, liu:et.al:2018:KDD,wu:et.al:2019:AAAI,wang:et.al:2019:SIGIR}, we set the vector dimension $d=100$ and the mini-batch size to 100 for the three datasets. We adopt the Adam optimizer with the initial learning rate 0.01, with a decay rate of 0.1 after every 3 epoches. We set the L2 penalty to $10^{-5}$ to avoid overfitting. We set $k=3$ for the item graph construction. We set the GSN iteration $T=4$. The hyper-parameters of the number of anchors $M$ and the number of GSN layer $L$ will be discussed in our later section~\ref{Sec:Hyperparameter} of hyper-parameter analysis.

\begin{table*}[t]
	\centering
	\caption{Overall performance comparison.}
	\begin{threeparttable}
		\begin{tabular}{llcccccc}
			\toprule
			~ & \multirow{2}*{Methods} & \multicolumn{2}{c}{Yoochoose 1/64} & \multicolumn{2}{c}{RetailRocket} & \multicolumn{2}{c}{Digineitca} \\
			\cmidrule(r){3-4} \cmidrule(r){5-6} \cmidrule(r){7-8}
			~ & ~ & HR@20(\%) & MRR@20(\%) & HR@20(\%) & MRR@20(\%) & HR@20(\%) & MRR@20(\%) \\
			\cmidrule(r){1-8}
			\multirow{5}*{Traditional Methods} & POP & 11.15 & 2.99 & 1.61 & 0.38 & 0.74 & 0.22 \\
			~ & S-POP & 39.62 & 18.99 & 39.53 & 27.17 & 21.60 & 13.51 \\
			~ & FPMC & 45.62 & 15.01 & 32.37 & 13.82 & 26.53 & 6.95 \\
			~ & SKNN & 63.77 & 25.22 & \underline{54.28} & 24.46 & 48.06 & 16.95 \\
			~ & STAN & 69.45 & 28.74 & 53.48 & 26.81 & 49.93 & 17.59 \\
			\cmidrule(r){1-8}
			\multirow{4}*{Neural Methods} & NARM & 68.32 & 28.63 & 50.22 & 24.59 & 49.70 & 16.17 \\
			~ & STAMP & 68.74 & 29.67 & 50.96 & 25.17 & 45.64 & 14.32 \\
			~ & CSRM & 69.85 & 29.71 & - & - & 51.69 & 16.92 \\
			~ & CoSAN & 70.04 & 29.85 & 52.47 & 24.40 & 48.34 & 15.22 \\
			\cmidrule(r){1-8}
			\multirow{6}*{Graph-based Methods} & SR-GNN & 70.57 & 30.94 & 50.32 & 26.57 & 50.73 & 17.59 \\
			~ & GC-SAN & 70.66 & 30.04 & 51.18 & \underline{27.40} & 50.84 & 17.79 \\
			~ & TAGNN & 71.02 & 31.12 & - & - & 51.31 & 18.03 \\
			~ & FLCSP & 71.58 & 31.31 & 52.63 & 25.85 & 51.68 & 17.27 \\
			~ & Disen-GNN & 71.46 & \underline{31.36} & - & - & 53.79 & 18.99 \\
			~ & GCE-GNN & \underline{72.18} & 30.84 & - & - & \underline{54.22} & \underline{19.04} \\
			~ & $S^2$-DHCN & 68.34 & 27.89 & 53.66 & 27.30 & 53.18 & 18.44 \\
			\cmidrule(r){1-8}
			\multirow{2}*{Proposed Methods} & GSN-IAS & \textbf{72.34} & \textbf{31.45} & \textbf{57.13} & \textbf{29.97} & \textbf{55.65} & \textbf{19.24} \\
			~ & Improv.(\%) & 0.22 & 0.29 & 5.25 & 9.38 & 2.64 & 1.05 \\
			\bottomrule
		\end{tabular}
		\begin{tablenotes}
			\centering
			\item[1] We have run 3 times in each dataset, and report the mean performances metrics, and the variances of different runs are consistently smaller than 0.01\% in our model.
		\end{tablenotes}
	\end{threeparttable}
	\label{Tble:Overall Performance}
\end{table*}

\subsection{Competitors}
We compare GSN-IAS with the following competitors, which are divided into three groups:

\par
\textbf{Traditional methods:} these apply some simple recommendation rules, like item popularity, matrix factorization, Markov chain or k-nearest-neighbors.

\par\noindent
$\cdot$ \textsf{POP} and \textsf{S-POP} recommend the most popular items in the training dataset and in the current session respectively.
\par\noindent
$\cdot$ \textsf{FPMC}~\cite{rendle:et.al:2010:WWW} combines matrix factorization and Markov chain for the next item recommendation.
\par\noindent
$\cdot$ \textsf{SKNN}~\cite{jannach:et.al:2017:Recsys} uses the k-nearest-neighbors (KNN) approach for session-based recommendation.
\par\noindent
$\cdot$ \textsf{STAN}~\cite{garg:et.al:2019:SIGIR} extends the SKNN method to incorporate sequential and temporal information with a decay factor.

\par
\textbf{Neural Network Models:} these focus on designing neural network to learn session representation.
\par\noindent
$\cdot$ \textsf{NARM}~\cite{li:et.al:2017:CIKM} incorporates a GRU layer with attention mechanism to learn session sequential representation.
\par\noindent
$\cdot$ \textsf{STAMP}~\cite{liu:et.al:2018:KDD} uses an attention mechanism to capture general interests of a session and current interests of the last click.
\par\noindent
$\cdot$ \textsf{CSRM}~\cite{wang:et.al:2019:SIGIR} considers the collaborative neighborhors of the latest $m$ sessions for predicting the sessision intent.
\par\noindent
$\cdot$  \textsf{CoSAN}~\cite{luo:et.al:2020:IJCAI} learns the session representation and predicts the session intent by investigating neighborhood sessions.

\par
\textbf{Graph-based Models:} these incorporate complex transition relations, co-occurrence relations into graph to capture richer information.
\par\noindent
$\cdot$  \textsf{SR-GNN}~\cite{wu:et.al:2019:AAAI} is the first work proposing to apply the graph neural network to learn item embedding.
\par\noindent
$\cdot$  \textsf{GC-SAN}~\cite{xu:et.al:2019:IJCAI} combines the graph neural network and multi-layers self-attention network.
\par\noindent
$\cdot$  \textsf{TAGNN}~\cite{yu:et.al:2020:SIGIR} uses a target attention mechanism to learn session representation.
\par\noindent
$\cdot$  \textsf{FLCSP}~\cite{zhang:et.al:2021:InfoSci} makes a decision fusion of latent categorical prediction and sequential prediction.
\par\noindent
$\cdot$  \textsf{Disen-GNN}~\cite{li:et.al:2022:TKDE} proposes a disentangled graph neural network to capture the session purpose with the consideration of factor-level attention on each item.
\par\noindent
$\cdot$  \textsf{GCE-GNN}~\cite{wang:et.al:2020:SIGIR} constructs session graph and global graph to learn two levels of item embedding.
\par\noindent
$\cdot$  \textsf{$S^2$-DHCN}~\cite{xia:et.al:2021:AAAI} learns the inter- and intra-session information from two types of hypergraphs.

\section{Experiment Results}
\label{Sec:Experiment Results}
We adopt two metrics commonly used in the SBR task, i.e., HR@20 and MRR@20. \textbf{HR@20} (Hit Rate)  focuses on whether the desired item appearing in the recommended list, without considering the ranking of item in the recommended list. \textbf{MRR@20} (Mean Reciprocal Rank) is the average of reciprocal ranks of the correct recommended items in the recommendation list.


\subsection{Overall Comparison (RQ1)}
Table~\ref{Tble:Overall Performance} presents the overall performance comparison between our GSN-IAS and the competitors, where the best results in every column are boldfaced and the second results are underlined. It is observed that our GSN-IAS outperforms the others on all three datasets in terms of the highest HR@20 and MRR@20 on three datasets.

\par
In traditional methods, the \textsf{POP} and \textsf{S-POP} based on simple recommendation rules are not competitive, since they lack the ability to capture the item sequential dependency in sessions. The \textsf{FPMC} using Markov chain to model the sequential relations also performs poorly, since the strict independence assumption of Markov chain is inconsistent with the real situation of the SBR task. The \textsf{SKNN} and \textsf{STAN} achieve competitive performances compared with neural methods and graph-based methods, and the \textsf{SKNN} even reaches the second best of HR@20 on RetailRocket dataset. They both consider the influence of similar sessions, and the \textsf{STAN} designs two time decay functions to describe the recency of a past session and the chronological order of items in a session, respectively. These imply that the collaborative information of similar sessions and the sequential relations in a session are important for the SBR task.

\par
The neural methods, i.e. the \textsf{NARM}, \textsf{STAMP}, \textsf{CSRM} and \textsf{CoSAN}, generally achieve significant performance improvements over traditional methods, which reflects the powerful representation ability of neural networks. However, these methods focus on learning a session representation, but have neglected to exploit some potential relations in between items from different sessions. Notice that compared with the \textsf{NARM} and \textsf{STAMP}, the \textsf{CSRM} and \textsf{CoSAN} take consideration of the collaborative information of similar sessions to assist the representation learning of current session, which makes them get better performance.

\par
The graph-based methods, i.e. the \textsf{SR-GNN}, \textsf{GC-SAN}, \textsf{TAGNN}, \textsf{FLCSP}, \textsf{Disen-GNN}, \textsf{GCE-GNN} and \textsf{$S^2$-DHCN}, achieve the state-of-the-art performance, which indicates that modeling the potential relations in between items by a graph structure is helpful for item and senssion representation learning. The first five construct a directed graph to capture item transition relations; While the \textsf{GCE-GNN} and \textsf{$S^2$-DHCN} consider item co-occurrences via constructing an undirected graph and a hypergraph respectively. The \textsf{GCE-GNN} gets three second best performances including HR@20 on Yoochoose 1/64 and HR@20, MRR@20 on Diginetica, and \textsf{$S^2$-DHCN} also gets competitive performances on RetailRocket and Diginetica. These suggest the importance of item co-occurrence relations for item embedding learning.

\par
Our approach \textsf{GSN-IAS} outperforms all the state-of-the-art algorithms on the three datasets. Especially on the RetailRocket dataset, the improvements of HR@20 and MRR@20 are 5.25\% and 9.38\% compared with the second best. This suggests the effectiveness of using our GSN for item embedding learning and informative anchor selection for item encoding. Note that except the \textsf{GCE-GNN} on Diginetica, other competitors only achieve the second best performance of either HR@20 or MRR@20; While our \textsf{GSN-IAS} achieves the best on both metrics.


\begin{table}[t]
	\centering
	\caption{Ablation study results.}
	\resizebox{0.5\textwidth}{!}{
		\begin{tabular}{lcccccc}
			\toprule
			Dataset & \multicolumn{2}{c}{Yoochoose 1/64} & \multicolumn{2}{c}{RetailRocket} & \multicolumn{2}{c}{Digineitca} \\
			\cmidrule(r){1-7}
			Method & HR@20 & MRR@20 & HR@20 & MRR@20 & HR@20 & MRR@20 \\
			\cmidrule(r){1-7}
			GSN-Anchor & 70.77 & 29.61 & 53.18 & 27.16 & 53.64 & 18.23 \\
			GSN-Item & 72.16 & 31.33 & 56.87 & 29.70 & 55.50 & 19.20 \\
			GSN-IAS-AvgFuse & 72.27 & 31.35 & 56.67 & 29.66 & 55.42 & 19.13 \\
			GSN-IAS & \textbf{72.34} & \textbf{31.45} & \textbf{57.13} & \textbf{29.97} & \textbf{55.65} & \textbf{19.24} \\
			\cmidrule(r){1-7}
			$(\sigma(\omega^a),\sigma(\omega^b))$ & \multicolumn{2}{c}{(0.8834, 0.1125)} & \multicolumn{2}{c}{(0.8421, 0.1517)} & \multicolumn{2}{c}{(0.8062, 0.1850)} \\
			\bottomrule
		\end{tabular}
	}
	\label{Tble:AblationStudy}
\end{table}

\subsection{Ablation Study}
We conduct ablation experiments to examine the effectiveness of each component in our \textsf{GSN-IAS}. We develop the following three ablation algorithms:
\begin{itemize}
	\item \textsf{GSN-Anchor}: It only uses item encoding for session representation learning and prediction.
	\item \textsf{GSN-Item}: It only uses item embedding for session representation learning and prediction.
	\item \textsf{GSN-IAS-AvgFuse}: It uses both item embedding and item encoding for session representation learning, but uses an average pooling for decision fusion.
\end{itemize}

\par
Table~\ref{Tble:AblationStudy} presents the results of ablation study. We first observe that \textsf{GSN-IAS} achieves the best performance than its ablation algorithms. Even the worst-performing \textsf{GSN-Anchor} has achieved the same level of performance as the state-of-the-art \textsf{SR-GNN}, \textsf{GC-SAN}, \textsf{TAGNN}, \textsf{FLCSP} methods, and \textsf{GSN-Item} also outperforms all competitors. These reflect two points: One is that each component in \textsf{GSN-IAS} makes its contribution to more accurate prediction, where \textsf{GSN-Item} plays the protagonist and \textsf{GSN-Anchor} the auxiliary; The other is that our GSN is an effective graph neural network to learn item embedding for the SBR task.

\par
From Table~\ref{Tble:AblationStudy}, we also observe that \textsf{GSN-IAS} outperforms the \textsf{GSN-IAS-AvgFuse}, which indicates the effectiveness of our adaptive fusion for learning fusion weights. On this basis, the last line of Table~\ref{Tble:AblationStudy} specifically shows the fusion weights learned for different datasets, where the fusion weights are learned differently for different datasets. Furthermore, it is not unexpected that $\sigma(\omega^a)>\sigma(\omega^b)$ on all datasets. This is in consistent with our previous analysis that the item embedding learning is the dominant factor, while the informative anchors information complements the former for the SBR task.


\begin{table}[t]
	\centering
	\caption{Performance comparison for using GCN, GAT, LightGCN for item embedding learning.}
	\resizebox{0.5\textwidth}{!}{
		\begin{tabular}{lcccccc}
			\toprule
			Dataset & \multicolumn{2}{c}{Yoochoose 1/64} & \multicolumn{2}{c}{RetailRocket} & \multicolumn{2}{c}{Digineitca} \\
			\cmidrule(r){1-7}
			Method & HR@20 & MRR@20 & HR@20 & MRR@20 & HR@20 & MRR@20 \\
			\cmidrule(r){1-7}
			GCN-IAS & 71.86 & 30.42 & 56.73 & 27.47 & 50.10 & 14.88 \\
			GAT-IAS & 72.16 & 30.81 & 57.05 & 29.51 & 55.06 & 18.61 \\
			LightGCN-IAS & \underline{72.47} & \textbf{30.98} & \textbf{57.18} & \underline{29.55} & \underline{55.64} & \underline{19.20} \\
			GSN-IAS & \textbf{72.34} & \underline{31.45} & \underline{57.13} & \textbf{29.97} & \textbf{55.65} & \textbf{19.24} \\
			\bottomrule
	\end{tabular}
}
	\label{Tble:Comparison GNNs}
\end{table}

\begin{figure*}[t]
	\centering
	\includegraphics[width=\textwidth, height=0.7\textwidth]{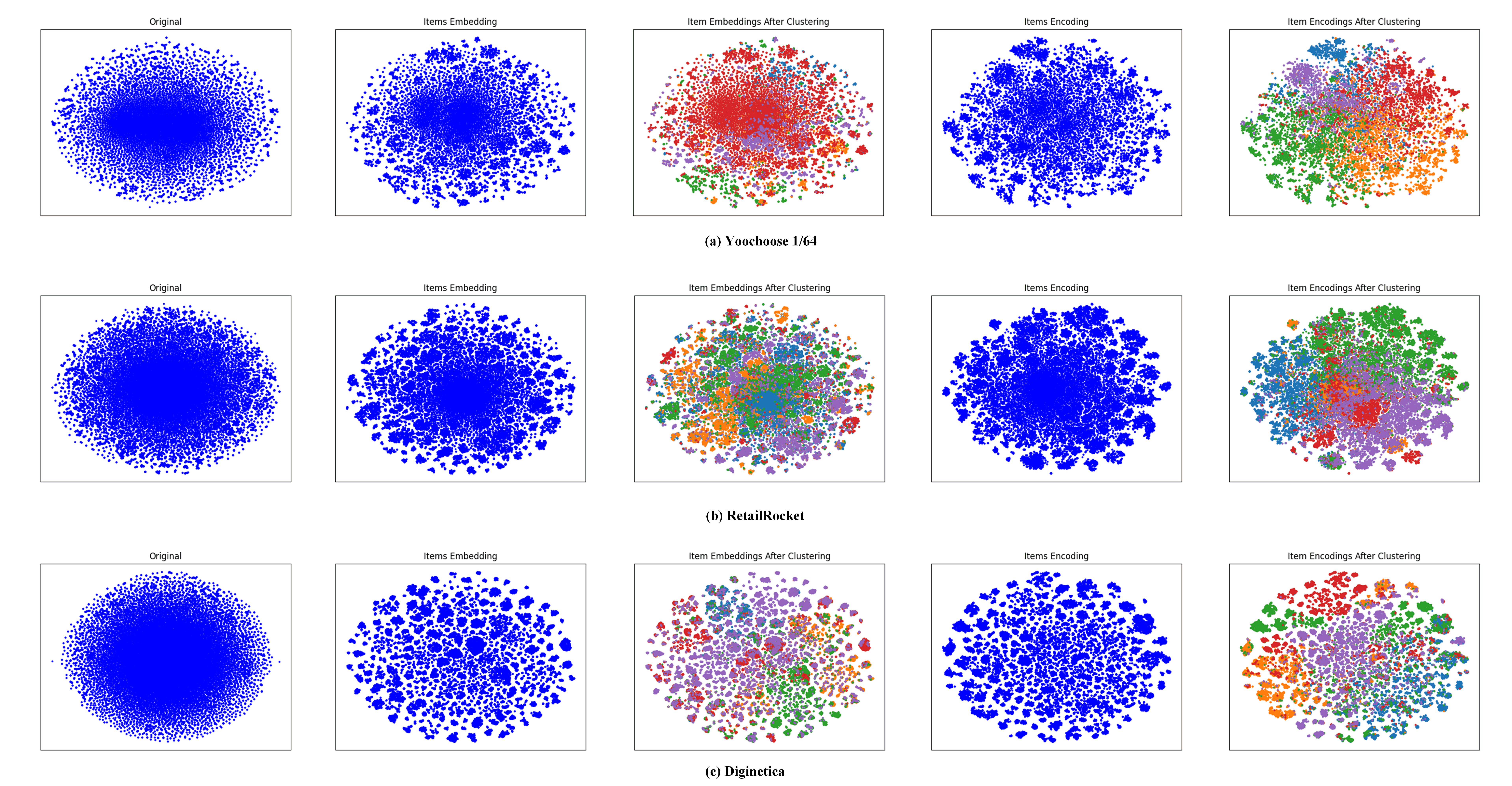}
	\caption{Visualization by t-sne for item embedding and anchor-based item embedding. We conduct a KMeans (K=5) clustering, and present the embeddings clustered in one class with the same color.}
	\label{Fig:Visualization}
\end{figure*}

\subsection{Comparison with other graph neural networks}
We conduct experiments using other well-knwon graph neural networks, including GCN~\cite{kipf:et.al:2017:ICLR}, GAT~\cite{velivckovic:et.al:2018:ICLR}, LightGCN~\cite{he:et.al:2020:SIGIR}, to replace our GSN for item embedding learning. We fix the hyper-parameter of the GNN layer $L=\{2,2,5\}$ and the number of anchors $M=\{100, 500, 1000\}$ for Yoochoose 1/64, RetailRocket, Dignetica, respectively. These comparison algorithms are denoted as \textsf{GCN-IAS}, \textsf{GAT-IAS} and \textsf{LightGCN-IAS}.

\par
Table~\ref{Tble:Comparison GNNs} presents the experiment results of using different GNNs. It is observed that \textsf{GSN-IAS} and \textsf{LightGCN-IAS} achieve comparable performance and are better than the other two, which validates our motivation of designing GSN without trainable parameters for ID-based item embedding learning in the SBR task. In addition, the \textsf{GAT-IAS} achieves better performance than that of \textsf{GCN-IAS}, which implies that learning aggregation weights based on node embedding is better than computing on graph adjacent matrix. Although the performance of our GSN is comparable to LightGCN, our GSN provides a new idea using a more simple, convenient and effective way to learn node embedding on a graph without using any additional parameters.

\begin{figure*}[t]
	\centering
	\includegraphics[width=\textwidth]{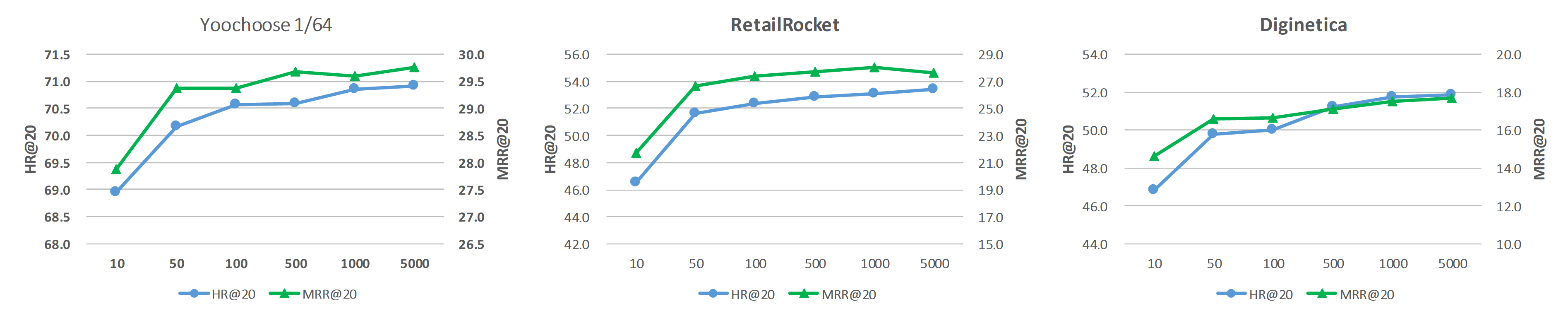}
	\caption{The performance of GSN-Anchor for using different numbers of selected anchors.}
	\label{Fig:NumAnchors}
\end{figure*}

\subsection{Visualization of item embedding}
Fig.~\ref{Fig:Visualization} presents the visualization of item embedding and item encoding in our \textsf{GSN-IAS} by the t-sne~\cite{van:et.al:2008:JMLR} algorithm. We observe that the results are consistent across different datasets. The original embeddings (the 1st column) are randomly initialized without much differentiations. Contrastively, after GSN learning for item embedding and after informative anchor selection for item encoding, the item embeddings (the 2nd column) and item encodings (the 4th column) reflect obvious clustering effect, that is, the similar items are closer to each other in the latent embedding space. To present the results more intuitively, we employ KMeans (K=5) algorithm to cluster the item embeddings (the 3rd column) and the item encodings (the 5th column). We can observe that the clustering effect is more obvious. Such clustering effect implies that both item embeddings and encodings can provide good item differentiations. Previous experiment results also reflect that such item differentiations can help to achieve more accurate recommendation.

\begin{table}[t]
	\centering
	\caption{Performance of using different numbers of GSN layers.}
	\resizebox{0.5\textwidth}{!}{
		\begin{tabular}{lcccccc}
			\toprule
			Dataset & \multicolumn{2}{c}{Yoochoose 1/64} & \multicolumn{2}{c}{RetailRocket} & \multicolumn{2}{c}{Digineitca} \\
			\cmidrule(r){1-7}
			Method & HR@20 & MRR@20 & HR@20 & MRR@20 & HR@20 & MRR@20 \\
			\cmidrule(r){1-7}
			GSN-IAS-1hop & 72.23 & 31.38 & 56.32 & 29.79 & 54.93 & 19.14 \\
			GSN-IAS-2hop & 72.34 & \textbf{31.45} & 57.13 & \textbf{29.97} & 55.21 & 19.22 \\
			GSN-IAS-3hop & 72.42 & 31.11 & 57.00 & 29.66 & 55.53 & 19.18 \\
			GSN-IAS-4hop & \textbf{72.52} & 31.07 & 57.34 & 29.63 & 55.57 & 19.13 \\
			GSN-IAS-5hop & 72.46 & 31.03 & \textbf{57.59} & 29.42 & \textbf{55.65} & \textbf{19.24} \\
			\bottomrule
		\end{tabular}
	}
	\label{Tble:NumLayers}
\end{table}

\subsection{Hyperparameter Analysis}
\label{Sec:Hyperparameter}
We finally examine the two hyperparameters in our \textsf{GSN-IAS}: One is the number of anchors, i.e. $M$; and the other is the number of stacked GSN layers, i.e $L$.

\par
Fig.~\ref{Fig:NumAnchors} presents the performance of GSN-Anchor for using different numbers of anchors. We observe that the trend of curve of HR@20 or MRR@20 on all the three datasets are consistent, that is, increasing first, then maintaining relatively stable, even decreasing slightly on RetailRocket with the increase of $M$. These results first indicate that too few anchors carry insufficient information; With increasing the number of anchors, the information provided by anchors tend to be saturated; While too many anchors may introduce some noise or repeated information to weaken the item differentiations.

\par
Table~\ref{Tble:NumLayers} presents the performance of using different layers of GSN. above contrast experiments. We first note that even the \textsf{GSN-IAS-1hop} still outperforms the state-of-the-art methods. With the increase of $L$, the HR@20 increases, but the growth rate slows down and even decreases on Yoochoose 1/64. The MRR@20 increases first and then decreases on Yoochoose 1/64 and RetailRocket. These results are easy to understand: Using a lager $L$, higher-order information of each node in the item graph can be explored, but more high-order neighbors may blur a node's information. The choice of $L$ is also related to the dataset size. The Yoochoose 1/64 is smaller than RetailRocket and Diginetica; So using a small $L$ suffices for Yoochoose 1/64. The RetailRocket and Diginetica contain about 2.1 times and 2.5 times more items than Yoochoose 1/64, needing a larger $L$ to explore more higher-order neighbors.

\section{Conclusion}\label{Sec:Conlusion}
In this paper, we have proposed a novel GSN-IAS model for the SBR task. Our model explores item co-occurrence to construct an item graph, on which we propose a new GSN neural network to optimize neighborhood affinity in ID-based item embedding learning. We have also deigned an informative anchor selection strategy to select anchors and to encode potential relations of all nodes to such anchors. Furthermore, we have employed a shared GRU to learn two session representations for two predictions. We have also proposed an adaptive decision fusion mechanism to fuse the two predictions to output the final recommendation list. Experiments on three public datasets have validated the superiority of our GSN-IAS model over the state-of-the-art algorithms.

\par
In this paper, we have especially discussed the characteristics of ID-based item embedding learning for the SBR task. We note that our GSN model may not only be suitable for such ID embedding learning in the SBR task, but may also be of further applications in other scenarios, like collaborative filtering. Our future work will explore the GSN potentials in other tasks. Furthermore, we suggest to further mine and encode some latent transition knowledge such as diverse transitional modes in between different kinds of items as another interesting future work.


%


%
%

\ifCLASSOPTIONcaptionsoff
  \newpage
\fi



%

\bibliographystyle{IEEEtran}
\bibliography{reference}

\appendix
\subsection{A Proof of Convergence of GSN}
\label{Appendix:GSNProof}
We investigate whether the core operation of GSN can converge after a sufficient number of iterations. We take an arbitrary node as example and conduct convergence analysis under the EM algorithm framework. We prove that the GSN iteration process is equivalent to an expectation-maximization (EM) algorithm estimating parameter of the vMF distribution by maximizing the likelihood probability $P(\mathbf{h}_i|\mathbf{c}_i)$.

\par
We first briefly review the EM algorithm to clarify our objective. The EM algorithm is an iterative method to find the maximum likelihood estimates of parameters for a statistical model depending on unobserved latent variables. The EM iteration alternates between performing an expectation step (E-step) to compute the expectation function of the log-likelihood, and a maximization step (M-step) to find such parameters maximizing the expected log-likelihood computed on the E-step. We can also derive the expectation function for our task.

\par
We assume that the parameter after $t$ iterations is $\mathbf{c}_i^{(t)}$, and EM algorithm optimizes the new parameter estimate $\mathbf{c}_i$ to increase the log-likelihood denoted as $L(\mathbf{c}_i)=\log (P(\mathbf{h}|\mathbf{c}_i))$, that is, $L(\mathbf{c}_i) > L(\mathbf{c}_i^{(t)})$. Let us consider the difference between them:
{\small
	\begin{eqnarray}
		\label{Eq:likelihood}
			& L(\mathbf{c}_i) - L(\mathbf{c}_i^{(t)}) = \log \left(\sum_{\bm{\alpha}}P(\mathbf{h},\bm{\alpha}|\mathbf{c}_i)\right) - \log P(\mathbf{h}|\mathbf{c}_i^{(t)}) \nonumber \\	
			& = \log \left(\sum_{\bm{\alpha}}P(\bm{\alpha}|\mathbf{h},\mathbf{c}_i^{(t)})\frac{P(\mathbf{h}|\bm{\alpha},\mathbf{c}_i)P(\bm{\alpha}|\mathbf{c}_i)}{P(\bm{\alpha}|\mathbf{h},\mathbf{c}_i^{(t)})}\right) - \log P(\mathbf{h}|\mathbf{c}_i^{(t)})  \nonumber  \\
			& \geq \sum_{\bm{\alpha}}P(\bm{\alpha}|\mathbf{h},\mathbf{c}_i^{(t)})\log\frac{P(\mathbf{h}|\bm{\alpha},\mathbf{c}_i)P(\bm{\alpha}|\mathbf{c}_i)}{P(\bm{\alpha}|\mathbf{h},\mathbf{c}_i^{(t)})} - \log P(\mathbf{h}|\mathbf{c}_i^{(t)})  \nonumber  \\
			& = \sum_{\bm{\alpha}}P(\bm{\alpha}|\mathbf{h},\mathbf{c}_i^{(t)})\log \frac{P(\mathbf{h}|\bm{\alpha},\mathbf{c}_i)P(\bm{\alpha}|\mathbf{c}_i)}{P(\bm{\alpha}|\mathbf{h},\mathbf{c}_i^{(t)})P(\mathbf{h}|\mathbf{c}_i^{(t)})}
\end{eqnarray}}
where $\bm{\alpha}$ are the latent variables. We let:
\begin{equation}
	\begin{split}
		& B(\mathbf{c}_i,\mathbf{c}_i^{(t)}) \triangleq L(\mathbf{c}_i^{(t)}) + \\
		& \sum_{\bm{\alpha}}P(\bm{\alpha}|\mathbf{h},\mathbf{c}_i^{(t)})\log \frac{P(\mathbf{h}|\bm{\alpha},\mathbf{c}_i)P(\bm{\alpha}|\mathbf{c}_i)}{P(\bm{\alpha}|\mathbf{h},\mathbf{c}_i^{(t)})P(\mathbf{h}|\mathbf{c}_i^{(t)})}
	\end{split}.
\end{equation}
So we have
\begin{equation}
	L(\mathbf{c}_i) \geq B(\mathbf{c}_i,\mathbf{c}_i^{(t)}).
\end{equation}
Obviously, $B(\mathbf{c}_i,\mathbf{c}_i^{(t)})$ is a lower boundary of $L(\mathbf{c}_i)$, where $L(\mathbf{c}_i^{(t)})=B(\mathbf{c}_i^{(t)},\mathbf{c}_i^{(t)})$. In order for $L(\mathbf{c}_i)$ to increase as much as possible, we choose $\mathbf{c}_i^{(t+1)}$ to maximize $B(\mathbf{c}_i,\mathbf{c}_i^{(t)})$ as follows:
{\small
	\begin{eqnarray}
			& & \mathbf{c}_i^{(t+1)} = \mathop{\arg\max}\limits_{\mathbf{c}_i} B(\mathbf{c}_i,\mathbf{c}_i^{(t)}) \\
			& = &\mathop{\arg\max}\limits_{\mathbf{c}_i} (L(\mathbf{c}_i^{(t)}) \nonumber \\
        & & \quad \quad  +  \sum_{\bm{\alpha}}P(\bm{\alpha}|\mathbf{h},\mathbf{c}_i^{(t)})\log \frac{P(\mathbf{h}|\bm{\alpha},\mathbf{c}_i)P(\bm{\alpha}|\mathbf{c}_i)}{P(\bm{\alpha}|\mathbf{h},\mathbf{c}_i^{(t)})P(\mathbf{h}|\mathbf{c}_i^{(t)})} ) \\
			& = & \mathop{\arg\max}\limits_{\mathbf{c}_i}(\sum_{\bm{\alpha}}P(\bm{\alpha}|\mathbf{h},\mathbf{c}_i^{(t)})\log(P(\mathbf{h}|\bm{\alpha},\mathbf{c}_i)P(\bm{\alpha}|\mathbf{c}_i))) \\
			& = & \mathop{\arg\max}\limits_{\mathbf{c}_i} (\sum_{\bm{\alpha}}P(\bm{\alpha}|\mathbf{h},\mathbf{c}_i^{(t)})\log P(\mathbf{h},\bm{\alpha}|\mathbf{c}_i) ) \\
			& = & \mathop{\arg\max}\limits_{\mathbf{c}_i} \mathbb{E}_{\bm{\alpha}}[\log P(\mathbf{h},\bm{\alpha}|\mathbf{c}_i)|\mathbf{h},\mathbf{c}_i^{(t)}]
\end{eqnarray}}
The $Q$-function is defined as follows.
\begin{definition}
	\textbf{$\bm{Q}$-function} is the expectation of the log-likelihood function $\log P(\mathbf{h},\bm{\alpha}|\mathbf{c}_i)$ with respect to the conditional probability $P(\bm{\alpha}|\mathbf{h},\mathbf{c}_i^{(t)})$ of the unobserved latent variable $\bm{\alpha}$ given observed data $\mathbf{h}$ and the current parameter $\mathbf{c}_i^{(t)}$, that is,
	\begin{equation}
		Q(\mathbf{c}_i,\mathbf{c}_i^{(t)}) = \mathbb{E}_{\bm{\alpha}}[\log P(\mathbf{h},\bm{\alpha}|\mathbf{c}_i)|\mathbf{h},\mathbf{c}_i^{(t)}].
	\end{equation}
\end{definition}
We can intuitively find that choosing parameter $\mathbf{c}_i$ to maximize $Q(\mathbf{c}_i,\mathbf{c}_i^{(t)})$ is equivalent to maximize $B(\mathbf{c}_i,\mathbf{c}_i^{(t)})$, which can increase the likelihood $L(\mathbf{c}_i)$. The $Q$-function is the core of EM algorithm, and the E-step and M-step are as follows:
\begin{itemize}
	\item E-step: Compute the $Q$-function for the given current parameter $\mathbf{c}_i^{(t)}$.
	\item M-step: Find the parameter $\mathbf{c}_i$ to maximize the $Q$-function.
\end{itemize}
We leverage the convergence of EM algorithm to discuss the convergence of our GSN core operation. According to the probability density function of vMF distribution, we have
\begin{align}
	& P(\mathbf{h}_i|\mathbf{c}_i) \propto \exp(\mathbf{h}_i^\mathsf{T}\mathbf{c}_i) \\
	& P(\mathbf{h}_j, \alpha_{ij}|\mathbf{h}_j) \propto \exp(\alpha_{ij}\mathbf{h}_j^\mathsf{T}\mathbf{c}_i)
\end{align}
where $\alpha_{ij}$ is the concentration parameter seen as a latent variable. We can write the log-likelihood in $Q$-function as follows:
\begin{eqnarray}
		& &\log P(\mathbf{h},\bm{\alpha}|\mathbf{c}_i) = \log (P(\mathbf{h}_i|\mathbf{c}_i) \prod_{v_j\in\mathcal{N}_i}P(\mathbf{h}_j,\alpha_{ij}|\mathbf{c}_i)) \nonumber \\
		&  \propto & \log (\exp (\mathbf{h}_i^{\mathsf{T}}\mathbf{c}_i)\prod_{v_j\in\mathcal{N}_i} \exp(\alpha_{ij}\mathbf{h}_j^\mathsf{T}\mathbf{c}_i) ) \\
		& = & \log (\exp(\mathbf{h}_i^\mathsf{T}\mathbf{c}_i + \sum_{v_j \in \mathcal{N}_i}\alpha_{ij}\mathbf{h}_j^\mathsf{T}\mathbf{c}_i)) \\
		& = & \mathbf{h}_i^\mathsf{T}\mathbf{c}_i + \sum_{v_j \in \mathcal{N}_i}\alpha_{ij}\mathbf{h}_j^\mathsf{T}\mathbf{c}_i
\end{eqnarray}
The $Q$-function can be written as:
\begin{eqnarray}
	\label{Eq:QFunction1}
		Q(\mathbf{c}_i,\mathbf{c}_i^{(t)}) &=& \mathbb{E}_{\bm{\alpha}}[\log P(\mathbf{h},\bm{\alpha}|\mathbf{c}_i)|\mathbf{h},\mathbf{c}_i^{(t)}] \\
		&=& \mathbb{E}_{\bm{\alpha}}[(\mathbf{h}_i^\mathsf{T}\mathbf{c}_i + \sum_{v_j \in \mathcal{N}_i}\alpha_{ij}\mathbf{h}_j^\mathsf{T}\mathbf{c}_i)|\mathbf{h},\mathbf{c}_i^{(t)}] \\
		&=& \mathbf{h}_i^\mathsf{T}\mathbf{c}_i + \sum_{v_j \in \mathcal{N}_i}\mathbb{E}(\alpha_{ij}|\mathbf{h},\mathbf{c}_i^{(t)})\mathbf{h}_j^\mathsf{T}\mathbf{c}_i
\end{eqnarray}

\par
We are going to compute $\mathbb{E}(\alpha_{ij}|\mathbf{h},\mathbf{c}_i^{(t)})$, denoted as $\hat{\alpha}_{ij}$, which is the expectation of latent variable $\alpha_{ij}$ given the current parameter $\mathbf{c}_i^{(t)}$ and observed data $\mathbf{h}$. Note that $\alpha_{ij}$ describes the concentration of $\mathbf{h}_j$ around the $\mathbf{c}_i$. We compute it based on the dot-product similarity and normalize it as a probability by the softmax function as follows:
\begin{eqnarray}
	\label{Eq:Ealpha}
		\hat{\alpha}_{ij} &=& \mathbb{E}(\alpha_{ij}|\mathbf{h},\mathbf{c}_i^{(t)}) \\
		&=& \frac{\exp (\mathbf{h}_j^{\mathsf{T}}\mathbf{c}_i^{(t)})}{\sum_{v_j \in \mathcal{N}_i}\exp (\mathbf{h}_j^{\mathsf{T}}\mathbf{c}_i^{(t)})}
\end{eqnarray}
Then the $Q$-function is rewritten as
\begin{equation}
	\label{Eq:QFunction2}
	Q(\mathbf{c}_i,\mathbf{c}_i^{(t)}) = \mathbf{h}_i^\mathsf{T}\mathbf{c}_i + \sum_{v_j \in \mathcal{N}_i}\hat{\alpha}_{ij}\mathbf{h}_j^\mathsf{T}\mathbf{c}_i.
\end{equation}
After the expectation $\mathbb{E}(\alpha_{ij}|\mathbf{h},\mathbf{c}_i^{(t)})$ has obtained, the $Q$-function is computed directly. This proves that the Eq.~\eqref{Eq:GSNWeight} is performing the E-step.

\par
After the E-step, we have obtained the expectation $\hat{\alpha}_{ij} = \mathbb{E}(\alpha_{ij}|\mathbf{h},\mathbf{c}_i^{(t)})$, and we would like to find parameter $\mathbf{c}_i$ to maximize $Q(\mathbf{c}_i,\mathbf{c}_i^{(t)})$, that is,
\begin{equation}
	\mathbf{c}_i^{(t+1)} = \mathop{\arg\max}\limits_{\mathbf{c}_i} Q(\mathbf{c}_i,\mathbf{c}_i^{(t)}).
\end{equation}
Note that we use an $\mathsf{UN}(\cdot)$ function to constraint $\mathbf{c}_i$, i.e. $\Vert \mathbf{c}_i \Vert_2 = 1$. This is actually a constrained extreme-value problem, which can be formulated as follows:
\begin{eqnarray}
		& \mathop{\max}\limits_{\mathbf{c}_i} Q(\mathbf{c}_i,\mathbf{c}_i^{(t)}) = \mathbf{h}_i^\mathsf{T}\mathbf{c}_i + \sum_{v_j \in \mathcal{N}_i}\hat{\alpha}_{ij}\mathbf{h}_j^\mathsf{T}\mathbf{c}_i \\
		& s.t. \quad \Vert \mathbf{c}_i \Vert_2^2 - 1 = 0
\end{eqnarray}
We use the Lagrange multipliers to solve this problem. The Lagrangian function can be written as follows:
\begin{equation}
	L(\mathbf{c}_i,\lambda,\eta) = \mathbf{h}_i^{\mathsf{T}}\mathbf{c}_i + \sum_{v_j \in \mathcal{N}_i}\hat{\alpha}_{ij}\mathbf{h}_j^{\mathsf{T}}\mathbf{c}_i+\lambda(\Vert \mathbf{c}_i \Vert_2^2 - 1),
\end{equation}
where $\lambda$ is the Lagrange multiplier. We take the derivative of the Lagrangian function with respect to $\mathbf{c}_i$ and $\lambda$, respectively, hence getting the extreme-value:
\begin{subnumcases}{\label{Eq:derivative}}
	\frac{\partial L}{\partial \mathbf{c}_i} = \mathbf{h}_i + \sum_{v_j \in \mathcal{N}_i}\hat{\alpha}_{ij}\mathbf{h}_j + 2\lambda\mathbf{c}_i = 0 & $ $ \label{Eq:derivative-ci} \\
	\frac{\partial L}{\partial \lambda} = \Vert \mathbf{c}_i \Vert_2^2 - 1 = 0 & $ $ \label{Eq:derivative-lambda}
\end{subnumcases}
A solution to these equations is as follows:
\begin{subnumcases}
		\mathbf{c}_i = -\frac{1}{2\lambda}(\mathbf{h}_i + \sum_{v_j \in \mathcal{N}_i}\hat{\alpha}_{ij}\mathbf{h}_j) & $ $\\
		\lambda  \neq 0 & $ $
\end{subnumcases}
The above solution shows that every $\mathbf{c}_i$ that satisfies the $\lambda \neq 0$ is an extreme point at which the $Q$-function takes an extreme value. We set $\lambda = -\frac{1}{2} \neq 0$, and $\mathbf{c}_i$ is computed by
\begin{equation}
	\mathbf{c}_i = \mathbf{h}_i + \sum_{v_j \in \mathcal{N}_i}\hat{\alpha}_{ij}\mathbf{h}_j
\end{equation}
This can directly prove that Eq.~\eqref{Eq:GSNUpdate} is performing the M-step.

\par
Let $\alpha_{ij}^{(t)}$ and $\mathbf{c}_i^{(t)}$ denote the result of the $t$-th E-step and M-step, respectively. According to Eq.~\eqref{Eq:likelihood}, we have
\begin{equation}
	L(\mathbf{c}_i^{(t-1)}) = \log (P(\mathbf{h}|\mathbf{c}_i^{(t-1)})) \leq \log (P(\mathbf{h}|\mathbf{c}_i^{(t)})) = L(\mathbf{c}_i^{(t)})
\end{equation}
The likelihood $L(\mathbf{c}_i)$ increases monotonically, and its upper-bound is zero. Therefore, the GSN convergence follows the convergence of the EM algorithm.

\subsection{Comparison with other graph neural networks}
\label{Appendix:GNNComparison}
We compare our GSN with two commonly used GNNs, i.e. LightGCN~\cite{he:et.al:2020:SIGIR} and GAT~\cite{velivckovic:et.al:2018:ICLR}.

\par
The core operation of LightGCN can be formulated by
\begin{align}
	\mathbf{H}^{(l+1)} &= \mathbf{D}^{-\frac{1}{2}} \mathbf{A} \mathbf{D}^{-\frac{1}{2}} \mathbf{H}^{(l)} \label{Eq:LightGCN-Matrix} \\
	\mathbf{h}_i^{(l+1)} &= \sum_{v_j \in \mathcal{N}_i} \frac{1}{\sqrt{| \mathcal{N}_i |}\sqrt{| \mathcal{N}_j |}}\mathbf{h}_j^{(l)} \label{Eq:LightGCN-Vector}
\end{align}
where Eq.~\eqref{Eq:LightGCN-Matrix} is the matrix form and Eq.~\eqref{Eq:LightGCN-Vector} is the vector form, and $\mathbf{L} = \mathbf{D}^{-\frac{1}{2}} \mathbf{A} \mathbf{D}^{-\frac{1}{2}}$ is the Laplacian matrix of the graph, $\mathcal{N}_j$ is the neighbor set of its neighbor node $v_j$. We can observe the commonalities and differences between GSN and LightGCN as follows:
\begin{itemize}
\item \textbf{Commonalities:} Both GSN and LightGCN do not involve a global transformation kernel, as well as no extra trainable parameters, which can ensure the training process only within the initial node embedding space.
\item \textbf{Differences:} The aggregation weights of LightGCN depends on numbers of one-order and two-order neighbors. The aggregation weights of GSN are computed as the embedding space similarities between a node and its neighbors.
\end{itemize}

\par
The core operation of GAT can be formulated by:
\begin{align}
	&\gamma_{ij} = \frac{\exp \left({ \Phi}(\mathbf{a}^{(l)\mathsf{T}}[\mathbf{W}^{(l)}\mathbf{h}_i^{(l)} || \mathbf{W}^{(l)}\mathbf{h}_j^{(l)}])\right)}{\sum_{v_j\in \mathcal{N}_i}\exp \left({\Phi}(\mathbf{a}^{(l)\mathsf{T}}[\mathbf{W}^{(l)}\mathbf{h}_i^{(l)} || \mathbf{W}^{(l)}\mathbf{h}_j^{(l)}])\right)} \label{Eq:GAT-Weight} \\
	&\mathbf{h}_i^{(l+1)} = \sigma (\sum_{v_j \in \mathcal{N}_i}\gamma_{ij}\mathbf{W}^{(l)}\mathbf{h}_j^{(l)})
\end{align}
where $\Phi$ is an activation function, say for example the LeakyReLU function, $\mathbf{W}^{(l)}, \mathbf{a}^{(l)}$ are the trainable parameters of the $l$-th GAT layer. We can observe the commonalities and differences between GSN and GAT as follows:
\begin{itemize}
	\item \textbf{Commonalities:} Both GSN and GAT compute the aggregation weights based on nodes embeddings.
	\item \textbf{Differences:} The GAT designs an attention function with some trainable parameters to compute the aggregation weights only once. Our GSN estimates the aggregation weights iteratively without additional trainable parameters. The iteration process aims at finding a local optimal aggregation weight for a node based on its neighbors' current embeddings.
\end{itemize}

\par
We next conduct complexity analysis. We compare the space complexity and time complexity of GSN with LightGCN. To facilitate subsequent analysis, we summarize the graph statistics used in this paper in Table~\ref{Tble:Graph}.

\textbf{Space Complexity Analysis:}
\par
For GSN, we only store the neighbors information of each node and the node embedding matrix $\mathbf{H} \in \mathbb{R}^{N \times d}$. We leverage the neighbor sampling for each node to sample only a subset of its neighbors based on the edge weights. The space complexity of one layer GSN is $O(Nd+Nr)$, where $r$ is the number of sampled neighbors for each node.

\par
For LightGCN, besides the node embedding matrix, LightGCN needs to store the Laplacian matrix $\mathbf{L}=\mathbf{D}^{-\frac{1}{2}} \mathbf{A} \mathbf{D}^{-\frac{1}{2}} \in \mathbb{R}^{N\times N}$ of the graph. We can use the sparse matrix storage method coordinate format to store the Laplacian matrix, which only stores the non-zero elements. The space complexity of LightGCN is $O(Nd + 2|E|)$, where $|E|$ is the number of edges in the graph.

\par
In this paper, we set $r=12$ for the three datasets, so $Nr=\{208,512;443,616;517,164\}$ are slightly less than $2|E|$, yet they are with the same magnitude order. As such, the space complexity of GSN is on par with the LightGCN.

\par
\textbf{Time Complexity Analysis: }
\par
For GSN, the core operations can be computed parallelly for all nodes. The time complexity of one layer GSN is $O(TNrd)$, where $T$ is the number of iterations, $N$ is the number of nodes in the graph, $d$ is the dimension of the node embedding. In general, the number of nodes in a graph is large, but $T, r$ are two small constants, that is, $N \gg r, N \gg T$. The time complexity can be reduced to $O(Nd)$.

\par
For LightGCN, it only computes the matrix multiplication between Laplacian matrix and feature matrix. Obviously, its time complexity is $O(|E|d)$. Therefore, the time complexity of GSN and LightGCN depends on the number of nodes and the number of edges in a graph, respectively.

\par
In this paper, from Table~\ref{Tble:Graph}, we can observe that the number of nodes is less than the edges, but they are still on the same magnitude order. As such, the time complexity of our GSN is on par with the LightGCN.

%

%





\end{document}